\documentclass[runningheads,11pt]{article}

\usepackage[margin=1in, letterpaper]{geometry}

\usepackage{amsthm}
\newtheorem{definition}{Definition}
\newtheorem{theorem}{Theorem}
\newtheorem{lemma}{Lemma}
\newtheorem{corollary}{Corollary}
\newtheorem{observation}{Observation}
\newtheorem{proposition}{Proposition}
\theoremstyle{remark}{\newtheorem{remark}{Remark}}

\usepackage{wrapfig}
\usepackage{authblk}

\usepackage{amsmath,amssymb,graphicx}
\usepackage{algorithm}
\usepackage[noend]{algpseudocode}
\usepackage{verbatim}
\usepackage{centernot,cancel}
\usepackage{xcolor}
\usepackage[colorinlistoftodos,prependcaption,textsize=tiny]{todonotes}
\usepackage{fancyhdr}
\usepackage{caption}
\usepackage{hyperref}

\usepackage{array}

\newcolumntype{L}[1]{>{\raggedright\let\newline\\\arraybackslash\hspace{0pt}}m{#1}}
\newcolumntype{C}[1]{>{\centering\let\newline\\\arraybackslash\hspace{0pt}}m{#1}}
\newcolumntype{R}[1]{>{\raggedleft\let\newline\\\arraybackslash\hspace{0pt}}m{#1}}

\usepackage{tabto}
\usepackage{makecell}
\usepackage{xcolor,colortbl}

\definecolor{Gray1}{gray}{0.7}
\definecolor{Gray2}{gray}{0.75}
\definecolor{Gray3}{gray}{0.80}
\definecolor{Gray4}{gray}{0.85}
\definecolor{Gray5}{gray}{0.87}
\definecolor{Gray6}{gray}{0.90}
\definecolor{Gray7}{gray}{0.92}
\definecolor{Gray8}{gray}{0.95}
\definecolor{Gray9}{gray}{0.97}
\definecolor{Gray10}{gray}{1}
\newcolumntype{[}{!{\vrule width 1pt}}

\title{Fault Tolerant Network Constructors
	\thanks{All authors were supported by the EEE/CS initiative NeST. The last author was also supported by the Leverhulme Research Centre for Functional Materials Design.
This work was partially supported by the EPSRC Grant EP/P02002X/1 on Algorithmic Aspects of Temporal Graphs.\newline
\textit{Email addresses:} \href{mailto:Othon.Michail@liverpool.ac.uk}{Othon.Michail@liverpool.ac.uk} (Othon Michail), \href{mailto:P.Spirakis@liverpool.ac.uk}{P.Spirakis@liverpool.ac.uk} (Paul G. Spirakis), \href{mailto:Michail.Theofilatos@liverpool.ac.uk}{Michail.Theofilatos@liverpool.ac.uk} (Michail Theofilatos)}
}

\author[1]{Othon Michail}
\author[1,2]{Paul G. Spirakis}
\author[1]{Michail Theofilatos}

\affil[1]{Department of Computer Science, University of Liverpool, UK} \affil[2]{Computer Engineering and Informatics Department, University of Patras, Greece}

\date{}

\begin{document}
\maketitle

\begin{abstract} \normalsize
	In this work, we consider adversarial crash faults of nodes in the network constructors model $[$Michail and Spirakis, 2016$]$.
	We first show that, without further assumptions, the class of graph languages that can be (stably) constructed under crash faults is non-empty but small.
	In particular, if an unbounded number of crash faults may occur, we prove that (i) the only constructible graph language is that of spanning cliques and (ii) a strong impossibility result holds even if the size of the graphs that the protocol outputs in populations of size $n$ need only grow with $n$ (the remaining nodes being \emph{waste}). 
	When there is a finite upper bound $f$ on the number of faults, we show that it is impossible to construct any \emph{non-hereditary} graph language and leave as an interesting open problem the \emph{hereditary case}. On the positive side, by relaxing our requirements we prove that: (i) permitting linear waste enables to construct on $n/(2f)-f$ nodes, any graph language that is constructible in the fault-free case,
	(ii) \emph{partial constructibility} (i.e., not having to generate all graphs in the language) allows the construction of a large class of graph languages.
	We then extend the original model with a minimal form of \emph{fault notifications}. 
	Our main result here is a \emph{fault-tolerant universal constructor}: We develop a fault-tolerant protocol for \emph{spanning line} and use it to simulate a linear-space Turing Machine $M$. This allows a fault-tolerant construction of any graph accepted by $M$ in linear space, with waste $min\{n/2 + f(n), \; n\}$, where $f(n)$ is the number of faults in the execution.	
	We then prove that increasing the permissible waste to $min\{2n/3 + f(n), \; n\}$ allows the construction of graphs accepted by an $O(n^2)$-space Turing Machine, which is asymptotically the maximum simulation space that we can hope for in this model.
	Finally, we show that logarithmic local memories can be exploited for a no-waste fault-tolerant simulation of any such protocol.
\end{abstract}

\noindent
\textbf{Keywords:} network construction; distributed protocol; self stabilization; fault tolerant protocol; dynamic graph formation; population; fairness; self-organization;\newline

\section{Introduction and Related Work}
\label{introduction}

In this work, we address the issue of the dynamic formation of graphs under faults. We do this in a minimal setting, that is, a population of agents running \textit{Population Protocols} that can additionally activate/deactivate links when they meet. This model, called \textit{Network Constructors}, was introduced in \cite{MS16a}, and is based on the \textit{Population Protocol} (PP) model \cite{AADFP06,AAER07} and the \textit{Mediated Population Protocol} (MPP) model \cite{MCS11-2}.
We are interested in answering questions like the following:
If one or more faults can affect the formation process, can we always re-stabilize to a correct graph, and if not, what is the class of graph languages for which there exists a fault-tolerant protocol?
What are the additional minimal assumptions that we need to make in order to find fault-tolerant protocols for a bigger class of languages?


Population Protocols run on networks that consist of computational entities called \textit{agents}. One of the challenging characteristics is that the agents have no control over the schedule of interactions with each other.
In a population of $n$ agents, repeatedly a pair of agents is chosen to interact. During an interaction their states are updated based on their previous states.
In general, the interactions are scheduled by a \textit{fair scheduler}. When the execution time of a protocol needs to be examined, a typical fair scheduler is the one that selects interactions uniformly at random.

Network Constructors (and its geometric variant \cite{Mi18}) is a theoretical model that may be viewed as a minimal model for programmable matter operating in a dynamic environment \cite{MS17}.
Programmable matter refers to any type of matter that can \textit{algorithmically} transform its physical properties, for example shape and connectivity. The transformation is the result of executing an underlying program, which can be either a centralized algorithm or a distributed protocol stored in the material itself.
There is a wide range of applications, spanning from distributed robotic systems \cite{GKR10}, to smart materials, and many theoretical models (see, e.g., \cite{DDGRS14,DDG18,MSS18,DFS19} and references therein), try to capture some aspects of them.

The main difference between PPs and Network Constructors is that in the PP (and the MPP) models, the focus is on computation of functions of some input values, while Network Constructors are mostly concerned with the stable formation of graphs that belong to some graph language. Fault tolerance must deal with the graph topology, thus, previous results on self-stabilizing PPs \cite{AAFJ08,BBB13,DFI17,CLV17} and MPPs \cite{MOKY12} do not apply here.

In \cite{MS16a}, Michail and Spirakis gave protocols for several basic network construction problems, and proved several universality results by presenting generic protocols that are capable of simulating a Turing Machine and exploiting it in order to stably construct a large class of networks, in the absence of crash failures.

In this work, we examine the setting where \textit{adversarial crash faults} may occur, and we address the question of which families of graph languages can be stably formed.
Here, adversarial crash faults mean that an adversary knows the rules of the protocol and can select some node to be removed from the population at any time. For simplicity, we assume that the faults can only happen sequentially. This means that in every step at most one fault may occur, as opposed to the case where many faults can occur during each step. These cases are equivalent in the Network Constructors model w.l.o.g., but not in the extended version of this model (which allows fault notifications) that we consider later. We discuss more about it in Section \ref{nNET}.

A main difference between our work and traditional self-stabilization approaches is that the nodes are supplied with constant local memory, while in principle they can form linear (in the population size) number of connections per node. 
Existing self-stabilization approaches that are based on restarting techniques cannot be directly applied here \cite{DIM93,Do00}, as the nodes cannot distinguish whether they still have some activated connections with the remaining nodes, after a fault has occurred. This difficulty is the reason why it is not sufficient to just reset the state of a node in case of a fault.
In addition, in contrast to previous self-stabilizing approaches \cite{GK10,DB01} that are based on \textit{shared memory} models, two adjacent nodes can only store $1$ bit of memory in the edge joining them, which denotes the existence or not of a connection between them.

	Angluin \textit{et al.} \cite{AAFJ08} incorporated the notion of self-stabilization into the population protocol model, giving self-stabilizing protocols for some fundamental tasks such as token passing and leader election. They focused on the goal of stably maintaining some property such as having a unique leader or a legal coloring of the communication graph.
	
	A previous work of Delporte-Gallet \textit{et al.} \cite{DFGR06} studies the issue of correctly computing functions on the node inputs in the Population Protocol model \cite{AADFP06}, in the presence of crash faults and transient faults that can corrupt the states of the nodes. They construct a transformation which makes any protocol that works in the failure-free setting, tolerant in the presence of such failures, as long as modifying a small number of inputs does not change the output.
	Guerraoui and Ruppert \cite{GR09} introduced an interesting model, called \textit{Community Protocol}, which extends the the Population Protocol model with unique identifiers and enough memory to store a constant number of other agents' identifiers. They show that this model can solve any decision problem in NSPACE($n\log{n}$) while tolerating a constant number of Byzantine failures.

	In \cite{Pe09}, Peleg studies logical structures, constructed over static graphs, that need to satisfy the same property on the resulting structure after node or edge failures.
	He distinguishes between the stronger type of fault-tolerance obtained for geometric graphs (termed \textit{rigid fault-tolerance}) and the more flexible type required for handling general graphs (termed \textit{competitive fault-tolerance}). It differs from our work, as we address the problem of constructing such structures over dynamic graphs.\\

\subsection{Our contribution}

The goal of any Network Constructor (NET) protocol is to stabilize to a graph that belongs to (or satisfies) some graph language $L$, starting from an initial configuration where all nodes are in the same state and all connections are disabled.
In \cite{MS16a}, only the fault-free case was considered.
In this work, we formally define the model that extends NETs allowing crash failures, and we examine protocols in the presence of such faults.
Whenever a node crashes, it is removed from the population, along with all its activated edges. This leaves the remaining population in a state where some actions may need to be taken by the protocol in order to eventually stabilize to a correct network.

We first study the constructive power of the original NET model in the presence of crash faults.
We show that the class of graph languages that is in principle constructible is non-empty but very small: for unbounded number of faults, we show that the only stably constructible language is the \textit{Spanning Clique}. We also prove a strong impossibility result, which holds even if the size of graphs  that the protocol outputs in populations of size $n$ need only grow with $n$ (the remaining nodes being \emph{waste}). For bounded number of faults, we show that any non-hereditary graph language is impossible to be constructed.
However, we show that by relaxing our requirements we can extend the class of constructible graph languages.
In particular, permitting linear waste enables to construct on $n/(2f) - f$ nodes, where $f$ is a finite upper bound on the number of faults, any graph language that is constructible in a failure-free setting.
Alternatively, by allowing our protocols to generate only a subset of all graphs in the language (partial constructibility), a large class of graph languages becomes constructible (see Section \ref{fault_tolerant_SNET}).



In light of the impossibilities in the Network Constructors model, we introduce the minimal additional assumption of \textit{fault notifications}.
In particular, after a fault on some node $u$ occurs, all the nodes that maintain an active edge with $u$ at that time (if any) are notified. If there are no such nodes, an arbitrary node in the population is notified.
In that way, we guarantee that at least one node in the population will sense the removal of $u$.
Nevertheless, we show some constructions that work without notifications in the case of a crash fault on an isolated node (Section \ref{nNET}).


We obtain two fault-tolerant universal constructors. One of the main technical tools that we use in them, is a fault-tolerant construction of a stable path topology (i.e., a line).
We show that this topology is capable of simulating a Turing Machine (abbreviated ``TM'' throughout this paper), and, in the event of a fault, is capable of always reinitializing its state correctly (Section \ref{universal_waste}).
Our protocols use a subset of the population (called \textit{waste}) in order to construct there a TM, while the graph which belongs to the required language is constructed in the rest of the population (called \textit{useful space}).
The idea is based on \cite{MS16a}, where they show several universality results by constructing on $k$ nodes of the population a network $G_1$ capable of simulating a TM, and then repeatedly drawing a random network $G_2$ on the remaining $n-k$ nodes.
The idea is to execute on $G_1$ the TM which decides the language $L$ with input the network $G_2$. If the TM accepts, it outputs $G_2$, otherwise the TM constructs a new random graph.


This allows a fault-tolerant construction of any graph accepted by a TM in linear space, with waste $min\{n/2 + f(n),n\}$, where $f(n)$ is the number of faults in the execution.
We finally prove that increasing the permissible waste to $min\{2n/3 + f(n), \; n\}$ allows the construction of graphs accepted by an $O(n^2)$-space Turing Machine, which is asymptotically the maximum simulation space that we can hope for in this model.


In order to give fault-tolerant protocols without waste, we design a protocol that can be composed in parallel with any protocol in order to make it fault-tolerant. The idea is to restart the protocol whenever a crash failure occurs. We show that restarting is impossible with constant local memory, if the nodes may form a linear (in the population size) number of connections; hence, to overcome this we supply the agents with logarithmic memory (Section \ref{nNET_restart}).

Finally, in Section \ref{conclusions} we conclude and discuss further interesting research directions opened by this work.

The following table summarizes all results proved in this paper.

\begin{center}
	
	\begin{tabular}{ |L{3.4cm}|L{3.4cm}||L{7.8cm}|  }
		\Xhline{2.1\arrayrulewidth}
		\multicolumn{3}{|c|}{\cellcolor{Gray5}Constructible languages} \\
		\hline
		\multicolumn{2}{|c||}{Without notifications} &
		\multicolumn{1}{c|}{With notifications}\\
		\Xhline{1.4\arrayrulewidth}
		
		\multicolumn{1}{|c|}{\cellcolor{Gray8}Unbounded faults} &
		\multicolumn{1}{c||}{\cellcolor{Gray8}Bounded faults} &
		\multicolumn{1}{c|}{\cellcolor{Gray8}Unbounded faults}\\
		
		\Xhline{2.5\arrayrulewidth}
		Only Spanning Clique &  Non-hereditary impossibility & Fault-tolerant protocols: Spanning Star, Cycle Cover, Spanning Line\\
		
		\Xhline{1.4\arrayrulewidth}
		Strong impossibility even with linear waste & A representation of any finite graph (partial constructibility) & Universal Fault-tolerant Constructors (with waste) \\
		\Xhline{1.4\arrayrulewidth}
		& Any constructible graph language with linear waste & Universal Fault-tolerant Restart (without waste) \\
		
		\Xhline{1.4\arrayrulewidth}
		
	\end{tabular}
	\captionof{table}{Summary of our results.}
\end{center}

\section{Model and Definitions}
\label{model}

A
Network Constructor (NET) is a distributed protocol defined by a 4-tuple $(Q, q_0, Q_{out}, \delta)$, where $Q$ is a finite set of node-states, $q_0 \in Q$ is the initial node-state, $Q_{out} \subseteq Q$ is the set of output node-states, and $\delta: Q \times Q \times \{0,1\} \rightarrow  Q \times Q \times \{0,1\}$ is the transition function.

In the generic case, there is an \textit{underlying interaction graph} $G_U = (V_U, E_U)$ specifying the permissible interactions between the nodes, and on top of $G_U$, there is a dynamic overlay graph $G_O=(V_O, E_O)$.
A mapping function $F$ maps every node in the overlay graph to a distinct underlay node.
In this work, $G_U$ is a \textit{complete undirected interaction graph}, i.e., $E_U = \{uv: u,v \in V_U$ and $u \neq v \}$, while the overlay graph consists of a population of $n$ initially \textit{isolated} nodes (also called \textit{processes} or \textit{agents}).

The NET protocol is stored in each node of the overlay network, thus, each node $u \in G_O$ is defined by a state $q \in Q$.
Additionally, each edge $e \in E_O$ is defined by a binary state (\textit{active/connected} or \textit{inactive/disconnected}).
Initially, all nodes are in the same state $q_0$ and all edges are inactive. The goal is for the nodes, after interacting and activating/deactivating edges for a while, to end up with a desired stable overlay graph, which belongs to some graph language $L$.

During a (pairwise) interaction, the nodes are allowed to access the state of their joining edge and either activate it ($\text{state} = 1$) or deactivate it ($\text{state} = 0$). When the edge state between two nodes $u, v \in G_O$ is activated, we say that $u$ and $v$ are \textit{connected}, or \textit{adjacent} at that time $t$, and we write $u \underset{t}{\sim} v$.


In this work, we present a version of this model that allows \textit{adversarial} \textit{crash failures}. A crash (or \textit{halting}) failure causes an agent to cease functioning and play no further role in the execution.
This means that all the adjacent edges of $F(u) \in G_U$ are removed from $E_U$, and, at the same time, all the adjacent edges of $u \in G_O$ become inactive.


The execution of a protocol proceeds in discrete steps. In every step, an edge $e \in E_U$ between two nodes $F(u)$ and $F(v)$ is selected by an \textit{adversary scheduler}, subject to some \textit{fairness} guarantee. The corresponding nodes $u$ and $v$ interact with each other and update their states and the state of the edge $uv \in G_O$ between them, according to a joint transition function $\delta$.
If two nodes in states $q_u$ and $q_v$ with the edge joining them in state $q_{uv}$ encounter each other, they can change into states $q_u'$, $q_v'$ and $q_{uv}'$, where $(q_u',q_v', q_{uv}') \in \delta(q_u,q_v, q_{uv})$.
In the original model, $G_U$ is the complete directed graph, which means that during an interaction, the interacting nodes have distinct roles. In our protocols, we consider a more restricted version, that is, \textit{symmetric} transition functions ($\delta(q_u,q_v,q_{uv}) = \delta(q_v,q_u,q_{uv})$), as we try to keep the model as minimal as possible.

A configuration is a mapping $C: V_I \cup E_I \rightarrow Q \cup \{0, 1\}$
specifying the state of each node and each edge of the interaction graph.
An execution of the protocol on input $I$ is a finite or infinite sequence of configurations, $C_0, C_1, C_2, \dots$, each of which is a set of states drawn from $Q \cup \{0,1\}$. In the initial configuration $C_0$, all nodes are in state $q_0$ and all edges are inactive.
Let $q_u$ and $q_v$ be the states of the nodes $u$ and $v$, and $q_{uv}$ denote the state of the edge joining them.
A configuration $C_k$ is obtained from $C_{k-1}$ by one of the following types of transitions:

\begin{enumerate}
	\itemsep0.5em
	\item \textbf{Ordinary transition:} $C_k = (C_{k-1}-\{q_u,q_v,q_{uv}\}) \cup \{q_u', q_v', q_{uv}'\}$ where $\{q_u, q_v, q_{uv}\} \subseteq C_{k-1}$ and $(q_u',q_v', q_{uv}') \in \delta(q_u, q_v, q_{uv})$.
	
	\item \textbf{Crash failure:} $C_k = C_{k-1} - \{q_u\} - \{q_{uv}: uv \in E_I\}$ where $\{q_u, q_{uv}\} \subseteq C_{k-1}$.
	
	
	
\end{enumerate}

We say that $C'$ is \textit{reachable from $C$} and write $C \rightsquigarrow C'$, if there is a sequence of configurations $C = C_0, C_1, \dots, C_t=C'$, such that $C_i \rightarrow C_{i+1}$ for all $i$, $ 0 \leqslant i < t$. The fairness condition that we impose on the scheduler is quite simple to state. Essentially, we do not allow the scheduler to avoid a possible step forever. More formally, if $C$ is a configuration that appears infinitely often in an execution, and $C \rightarrow C'$, then $C'$ must also appear infinitely often in the execution. Equivalently, we require that any configuration that is always reachable is eventually reached.

We define the output of a configuration $C$ as the graph $G(C) = (V,E)$ where $V = \{u \in V_O: C(u) \in Q_{out}\}$ and $E = \{uv: u,v \in V,\; u \neq v, \text{ and } C(uv) = 1\}$.
If there exists some step $t \geq 0$ such that $G_Ο(C_i) = G$ for all $i \geq t$, we say that the output of an execution $C_0, C_1, \dots$ \textit{stabilizes} (or \textit{converges}) to graph $G$, every configuration $C_i$, for $i \geq t$, is called \textit{output-stable}, and $t$ is called the \textit{running time} under our scheduler.
We say that a protocol $\Pi$ stabilizes eventually to a graph $G$ of \textit{type} $L$ if and only if after a finite number of pairwise interactions, the graph defined by 'on' edges does not change and belongs to the graph language $L$.


\begin{definition}\label{def:constructs}
We say that a protocol $\Pi$ \emph{constructs} a graph language $L$ if: (i) every execution of $\Pi$ on $n$ nodes stabilizes to a graph $G \in L$ s.t. $|V(G)|=n$ and (ii) $\forall  G \in L$ there is an execution of $\Pi$ on $|V(G)|$ nodes that stabilizes to $G$, where $|V(G)|$ is the degree of the graph $G$.
\end{definition}

\begin{definition}
We say that a protocol $\Pi$ \emph{partially constructs} a graph language $L$, if: (i) from Definition \ref{def:constructs} holds and (ii) $\exists G \in L$ s.t. no execution of $\Pi$ on $|V(G)|$ nodes stabilizes to $G$.
\end{definition}


\begin{definition}[Fault-tolerant protocol]
	Let $\Pi$ be a NET protocol that, in a failure-free setting, constructs a graph $G \in L$.
$\Pi$ is called \emph{{$f$-fault-tolerant}} if for any population size $n > f$, any execution of $\Pi$ constructs a graph $G \in L$, where $|V(G)| = n - f$.
We also call $\Pi$ \emph{{fault-tolerant}} if the same holds for any number $f \leq n-2$ of faults.
\end{definition}

\begin{definition}[Constructible language]\label{def:constructible_language}
A graph language $L$ is called \emph{constructible} (\emph{partially constructible}) if there is a protocol that constructs (partially constructs) it. Similarly, we call $L$ \emph{constructible under $f$ faults}, if there is an $f$-fault-tolerant protocol that constructs $L$, where $f$ is an upper bound on the maximum number of faults during an execution.
\end{definition}


\begin{definition}[Critical node]
	Let $G$ be a graph that belongs to a graph language $L$. Call $u$ a \emph{critical node} of $G$ if by removing $u$ and all its edges, the resulting graph $G' = G - \{u\} - \{uv: v \sim u \}$, does not belong to $L$.
	In other words, if there are no critical nodes in $G$, then any (induced) subgraph $G'$ of $G$ that can be obtained by removing nodes and all their edges, also belongs to $L$.
\end{definition}

\begin{definition}[Hereditary Language]\label{def:hereditary}
	A graph language $L$ is called Hereditary if for any graph $G \in L$, every induced subgraph of $G$ also belongs to $L$. In other words, there is no graph $G \in L$ with critical nodes.
\end{definition}

\noindent This notion is known in the literature as \textit{hereditary property} of a graph w.r.t. (with respect to) some graph language $L$.
Observe that if there exists a graph $G$ s.t. for any induced subgraph $G'$ of $G$, $G' \in L$, does not imply that the same holds for any graph in $L$.
For example, consider the graph language $L = \{G: G \text{ has an even number of edges}\}$ and a graph $G$ which consists of any number of connected pairs of nodes. Then, by removing any number of nodes, the number of the edges remains even. However, a different topology, such as a star with even number of edges does not have this property.

Some examples of Hereditary Languages are ``Bipartite graph'', ``Planar graph'', ``Forest of trees'', ``Clique'', ``Set of cliques'', ``Maximum node degree $\leq \Delta$'' and so on.

In this work, unless otherwise stated, a graph language $L$ is an infinite set of graphs satisfying the following properties:
\begin{enumerate}
	\item (\textit{No gaps}): For all $n \geq$ c, where $c \geq 2$ is a finite integer, $\exists G \in L$ of order $n$.
	\item (\textit{No Isolated Nodes}): $\forall G \in L$ and $\forall u \in V(G)$, it holds that $d(u) \geq 1$ (where $d(u)$ is the degree of $u$).
\end{enumerate}

\noindent Even though graph languages are not allowed to contain isolated nodes, there are cases in which a protocol might be allowed to output one or more isolated nodes.
In particular, if a protocol $\Pi$ constructing $L$ is allowed a waste of at most $w$, then whenever $\Pi$ is executed on $n$ nodes, it must output a graph $G \in L$ of order $|V(G)| \geq n-w$, leaving at most $w$ nodes in one or more separate components (could be all isolated).

\section{Network Constructors without Fault Notifications}\label{fault_tolerant_SNET}

In this section, we study the constructive power of the original NET model in the presence of crash faults.
We start from the case in which the number of nodes that crash during an execution can be anything from $0$ up to $n-2$ nodes.
We are interested in characterizing the class of constructible graph languages.
Observe that we cannot trivially conclude that the adversary can always leave us with just $2$ nodes, only allowing our protocols to form a line of length $1$.
This is because our definition of constructible languages under faults takes into account all possible executions with $f$ faults, for all values of $f \in \{0,1, \dots, n-2\}$.
We show that in the case where the number of faults cannot be bounded by a constant number, the only language that is constructible under any number of faults is the $L_c = \{G: G \text{ is a spanning clique}\}$.

We then consider the setting where only a constant number of faults can occur during an execution, and we show that there is no protocol that constructs any graph language $L$, tolerating even a single fault if $L$ is not Hereditary.
However, if we allow linear waste in the population, any language that is constructible without faults, is now constructible.

Finally, we show a family of graph languages that is partially constructible (without waste in the population). The exact characterization of the class of partially constructible languages remains as an open problem.

\subsection{Unbounded Number of Faults} \label{sec:unbounded}

In this section we consider the setting where the number of faults can be any number up to $n-2$, where $n$ is the number of nodes.
We first present a protocol which constructs the language $\textit{Spanning Clique} = \{G: G \text{ is a spanning clique}\}$, and we prove that this protocol can tolerate any number of faults.

Let \textit{Clique} be the following $2-$state \textit{symmetric} protocol.
If we consider the case where no crash faults are allowed, for any population size, \textit{Clique} Protocol stabilizes to a clique with all the nodes in state \textit{r}.

\begin{algorithm}[ht]
	\floatname{algorithm}{Protocol} 
	\caption{Clique}\label{protocol:clique}
	\begin{algorithmic}[1000]
		
		\State $Q = \{b,\; r\}$
		\State Initial state: $b$
		\State $ $
		\State $\delta:$

		\State $(b,\; b, \; 0) \rightarrow (b,\; r, \; 0)$		
		\State $(b,\; r, \; 0) \rightarrow (r,\; r, \; 0)$		
		\State $(r,\; r, \; 0) \rightarrow (r,\; r, \; 1)$
		
		\State $ $
		
		\State \textbackslash \textbackslash All transitions that do not appear have no effect.
		
	\end{algorithmic}
\end{algorithm}

\begin{lemma}\label{lemma:clique}
	\emph{Clique} (Protocol \ref{protocol:clique}) is a fault-tolerant protocol for \emph{Spanning Clique}.
\end{lemma}
\begin{proof}
	Let $f < n$ and assume that $f$ nodes crash during the execution. Call $S$ the remaining $n - f$ nodes.
	
	\noindent	(a) If all nodes in $S$ are in state $b$, then the remaining nodes shall form a clique (in state $r$).
	
	\noindent	(b) If all nodes in $S$ are in state $r$, then again, \textit{Clique Protocol} stabilizes to a clique.

	\noindent	(c) If $S$ contains both colors, then the $r-$nodes will convert the $b-$nodes to $r$ and again \textit{Clique Protocol} stabilizes to a clique.
\end{proof}

By Lemma \ref{lemma:clique}, we know that the language \textit{Spanning Clique} is constructible under $n-2$ faults. To clarify, this means that for any execution of Protocol \ref{protocol:clique} on $n$ nodes, $f$ of which crash ($f \in \{0,1, \dots, n-2\}$), Protocol \ref{protocol:clique} is guaranteed to stabilize to a clique of order $n-f$.

We will now prove that (due to the power of the adversary), no other graph language is constructible under unbounded faults.

\begin{lemma}\label{lemma:negative_unbounded}
	Let $\Pi$ be a protocol constructing a language $L$ and $G \in L$ a graph that $\Pi$ outputs on $|V(G)|$ nodes. If $G$ has an independent set $S \subseteq V$, s.t. $|S| \geq 2$, then there is an execution of $\Pi$ on $n$ nodes which stabilizes on $|S|$ isolated nodes (where $|S| = n-f$ and $f$ is the number of faults in that execution).
\end{lemma}
\begin{proof}
	Consider an execution of $\Pi$ that outputs $G$. By definition, there is a point in this execution after which no further edge updates can occur (no matter what the infinite execution suffix will be). Take any configuration $C_{stable}$ after that point and consider its sub-configuration $C_S$ induced by the independent set $S$. Observe that $C_S$ encodes the state of each node $u \in S$ in that particular stable configuration $C_{stable}$.
	Denote also by $Q_S$ the multiset of all states assigned by $C_S$ to the nodes in $S$.
	
	Every state in $Q_S$ is reachable (in the sense that there exists an execution that produces it).
	For each $q \in Q_S$ consider the smallest population $V_q$ in which there is some execution $a_q$ of protocol $\Pi$ that produces state $q$.
	Consider the population $V = \bigcup_{q \in Q_S} V_q$ (or equivalently of size $n = \sum_{q \in Q_S}{|V_q|}$).
	
	For each $V_q$ in population $V$ we execute $a_q$ until $q$ is produced on some node $u_q$. After this, every $q \in Q_S$ is present in the population $V$.
	Then, the adversary crashes all nodes in $V_q \setminus \{u_q\}$ (i.e., only $u_q$ remains alive in each $V_q$). This leaves the execution with a set of alive nodes equivalent in cardinality and configurations to the independent set $S$ under $C_S$.
	
	The above construction is a finite prefix of fair executions. For the sake of contradiction, assume that in any fair continuation of the above prefix, $\Pi$ eventually stabilizes to a graph with no isolated nodes (as required by the fact that $\Pi$ constructs a graph language $L$).
	Take one such continuation $\gamma$. As $\gamma$ starts from a configuration in all respects equivalent to that of $S$ under $C_{stable}$, it follows that $\gamma$ can also be applied to $C_{stable}$ and in particular on the independent set $S$ starting from $C_S$. It follows that $\gamma$ must have exactly the same effect as before, that is, eventually it will cause the activation of at least one edge between the nodes in $S$.
	But this violates the fact that $C_{stable}$ is a stable configuration, therefore no edge could have been activated by $\Pi$ in the continuation, implying that the continuation must have been an execution stabilizing on $|S|$ isolated nodes.
\end{proof}

\begin{theorem}\label{theorem:unbounded_1}
	Let $L$ be any graph language such that $L \neq \text{Spanning Clique}$. Then, there is no protocol that constructs $L$ if an unbounded number of crash failures may occur.
\end{theorem}
\begin{proof}
	As $L \neq \textit{Spanning Clique}$, there exists $G \in L$ such that $G$ is not complete (and by definition no $G' \in L$ has isolated nodes). Therefore $G$ has an independent set $S$ of size at least $2$. If there exists a protocol $\Pi$ that constructs $L$, then by Lemma \ref{lemma:negative_unbounded} there must be an execution of $\Pi$ which stabilizes on at least $2$ isolated nodes. The latter is a stable output not in $L$, therefore a contradiction.
\end{proof}

\begin{theorem}\label{theorem:unbounded_2}
	If an unbounded number of faults may occur, the Spanning Clique is the only constructible language.
\end{theorem}
\begin{proof}
	Directly from Lemma \ref{lemma:clique} and Theorem \ref{theorem:unbounded_1}.
\end{proof}

\begin{theorem}
	Let $L$ be any graph language such that the graphs $G\in L$ have maximum independent sets whose size grows with $|V(G)|$. If the useful space of protocols is required to grow with $n$, then there is no protocol that constructs $L$ in the unbounded-faults case.
\end{theorem}
\begin{proof}
The proof is a direct application of Lemma \ref{lemma:negative_unbounded}. As the size of the maximum independent set of $G$ grows with $|V(G)|$ in $L$, and the useful space is a non-constant function of $n$, it follows that, as $n$ grows, the stable output-graph (on the useful space) has an independent set of size that grows with $n$ (consider, for example, the leaves of binary trees of growing size as such a growing independent set). As any such stable independent set of size $g(n)$ implies that another execution has to stabilize to $g(n)$ isolated nodes, it follows that any protocol for $L$ would produce infinitely many stable outputs of isolated nodes. The latter is contradicting the fact that the protocol constructs $L$.   
\end{proof}

\subsection{Bounded Number of Faults}

The exact characterization established above, shows that when an unbounded number of crash faults may occur, we cannot hope for non-trivial constructions.
We now relax the power of the faults adversary, so that there is a \textit{finite upper bound} $f$ on the number of faults.
In particular, fixing any such $f \geq 0$ in advance, it is guaranteed that $\forall n \geq 0$ and all executions of a protocol on $n$ nodes, at most $f$ nodes may fail during the execution.
Then the class of constructible graph languages is naturally parameterized in $f$.
We first show that non-hereditary languages are not constructible under $1$ fault.

\begin{theorem} \label{theorem:critical_nodes}
	If there exists a critical node in $G$, there is no 1-fault-tolerant NET protocol that stabilizes to it.
\end{theorem}
\begin{proof}
	Let $\Pi$ be a NET protocol that constructs a graph language $L$, tolerating one crash failure.
	Consider an execution $E$ and a sequence of configurations $C_0, C_1, \; \dots$ of $E$.
	Assume a time $t$ that the output of $E$ has stabilized to a graph $G \in L$ (i.e., $G(C_i) = G$, $\forall i \geq t$).
	Let $u$ be a critical node in $G$. Assume that the scheduler removes $u$ and all its edges (crash failure) at time $t'>t$, resulting to a graph $G' \notin L$.
	In order to fix the graph (i.e., re-stabilize to a graph $G'' \in L$), the protocol must change at some point $t''$ the configuration. This can only be the result of a state update on some node $v$.
	Now, call $E'$ the execution that node $u$ does not crash and, besides that, is the same as $E$. Then, between $t'$ and $t''$ the node $v$ has the same interactions as in the previous case where node $u$ crashed.
	This results to the same state update in $v$, since it cannot distinguish $E$ from $E'$. The fact that $u$ either crashes or not, leads to the same result (i.e., $v$ tries to fix the graph thinking that $u$ has crashed).
	This means that if we are constantly trying to detect faults in order to deal with them, this would happen indefinitely and the protocol would never be stabilizing. Consider that the network has stabilized to $G$. At some point, because of the infinite execution, a node will surely but wrongly detect a crash failure. Thus, $G$ has not really stabilized.
\end{proof}

\noindent By Definition \ref{def:hereditary} and Theorem \ref{theorem:critical_nodes} it follows that.

\begin{corollary}
	If a graph language $L$ is non-hereditary, it is impossible to be constructed under a single fault.
\end{corollary}

\noindent Note that this does not imply that any Hereditary language is constructible under constant number of faults. We leave this as an interesting open problem. \\

On the positive side we show that in the case of bounded number of faults, there is a non-trivial class of languages that is partially constructible.
Consider the class of graph languages defined as follows. Any such language $L_{D,f}$ in the family is uniquely specified by a graph $D=([k],H)$ and the finite upper bound $f<k$ on the number of faults.
A graph $G=(V,E)$ belongs to $L_{D,f}$ iff there are $k$ partitions $V_1, V_2, \dots, V_k$ of $V$ s.t. for all $1 \leq i,j \leq k$, $||V_i| - |V_j|| \leq f+1$.
In addition, $E$ is constructed as follows.
The graph $D = ([k],H)$, possibly containing self-loops, defines a neighboring relation between the $k$ partitions.
For every $(i,j) \in H$ (where possibly $i=j$), $E$ contains all edges between partitions $V_i$ and $V_j$, i.e., a complete bipartite graph between them (or a clique in case $i=j$). As no isolated nodes are allowed, every $V_i$ must be fully connected to at least one $V_j$ (possibly itself).


We first consider the case where $k = 2^\delta$, for some constant $\delta \geq 0$, and we provide a protocol that divides the population into $k$ partitions.
The protocol works as follows:
initially, all nodes are in state $c_0$ (we call this the partition $0$).
When two nodes in states $c_i$, where $i \geq 0$ interact with each other, they update their states to $c_{2i+1}$ and $c_{2i+2}$, moving to partitions $2i+1$ and $2i+2$ respectively.
When $j = 2i+1 \geq k-1$ (or $j = 2i+2 \geq k-1)$ for the first time, it means that the node has reached its final partition. It updates its state to $P_m$, where $m = j - k + 1$, thus, the final partitions are $\{P_0, P_1, \dots, P_{k-1}\}$.

This process divides each partition into two partitions of equal size. However, in the case where the number of nodes is odd, a single node remains unmatched. For this reason, all nodes participate to the final formation of $H$ regardless of whether they have reached their final partitions or not. There is a straightforward mapping of each internal partition to a distinct leaf of the binary tree, that is, each partition $c_i$ behaves as if it were in partition $P_i$.
In order to avoid false connections between the partitions, we also allow the nodes to disconnect from each other if they move to a different partition.
This process guarantees that eventually all nodes end up in a single partition, and their connections are strictly described by $H$.

\begin{algorithm}[H]
	\floatname{algorithm}{Protocol}
	\caption{Graph of Supernodes}\label{protocol:partition}
	\begin{algorithmic}[1000]
		
		\State $Q = \{c_i, P_j\}$, $0 \leq i \leq 2(k-1)$, $0 \leq j \leq k-1$
			
		\State Initial state: $c_0$
		\State $ $
		
		\State $\delta:$
		
		\State \textbackslash \textbackslash Partitioning
		
		\State $1.\;(c_i, \; c_i, \; 0) \rightarrow (c_{2i+1}, \; c_{2i+2}, \; 0)$, if $(i+1)<k$
		
		\State $2.\;(c_i, \; \cdot, \; \cdot) \rightarrow (P_j, \; \cdot, \; \cdot)$, if $(i \geq k - 1)$, $j=i-k+1$

		\State $ $
		
		\State \textbackslash \textbackslash Formation of graph H
		\State $3.\;(P_i, P_j, 0) \rightarrow (P_i, P_j, 1)$, if $((i,j) \in H)$
		\State $4.\;(P_i, P_j, 1) \rightarrow (P_i, P_j, 0)$, if $((i,j) \notin H)$

		\State $5.\;(c_i, P_j, 0) \rightarrow (c_i, P_j, 1)$, if $((i,j) \in H)$
		\State $6.\;(c_i, P_j, 1) \rightarrow (c_i, P_j, 0)$, if $((i,j) \notin H)$

		\State $ $
		
		\State \textbackslash \textbackslash All transitions that do not appear have no effect.
		
	\end{algorithmic}
\end{algorithm}

\begin{lemma} \label{lemma:partition_correctness}
	In the absence of faults, Protocol \ref{protocol:partition}, divides the population into $k$ partitions of at least $n/k - 1$ nodes each.
\end{lemma}
\begin{proof}
	Initially all nodes are in state $c_0$. When two $c_0$ nodes interact with each other, one of them becomes $c_1$ and the other one $c_2$. This means that all $n$ nodes split into two partitions of equal size. No node can become $c_0$ again at any time during the execution. In addition, there is only one partition $c_j$ that produces nodes of some other partition $c_i$, where $i$ is either $2j+1$ or $2j+2$, and the size of them are half the size of $c_j$.
	This process can be viewed as traversing a labelled binary tree, until all nodes reach to their final partition. A node in state $c_{i}$ has reached its final partition when $i \geq k-1$.
	This process describes a subdivision of the nodes, where each partition splits into two partitions of equal size. 
	The final partitions are $\{c_{k-1}, c_{k}, \dots, c_{2k-2}\}$.
	
	Assume now that the initial population size is $n_0$ (level $0$ of the binary tree). If $n_0$ is even, the size of the following two partitions $c_1$ and $c_2$ will be $n_0/2$. If $n_0$ is odd, one node remains unmatched, thus, the size of $c_1$ and $c_2$ will be $n_1 = \frac{n_0-1}{2}$.
	In the next level of the binary tree, at most one node will remain unmatched in each partition, thus $n_2 = \frac{n_1-1}{2}$. Consequently, the size of a partition in level $p$ can be calculated recursively, and (in the worst case) it is $n_p = \frac{n_{p-1}-1}{2}$.

\begin{equation}
\begin{split}
	n_p &= \frac{n_{p-1}-1}{2} = \frac{n_{p-1}}{2} - \frac{1}{2} = \frac{\frac{n_{p-2}}{2}-\frac{1}{2}}{2} -\frac{1}{2} = \\
	&= \frac{n_{p-2}}{4} - \frac{1}{4} - \frac{1}{2} = \dots = \frac{n_0}{2^p} - \sum_{i=1}^{p}\frac{1}{2^i} = \frac{n_0}{2^p} - (1-2^{-p}) > \frac{n_0}{2^p} - 1
\end{split}
\end{equation}

\noindent For $p = \log{k}$ levels, each partition has either $\frac{n_0}{k}$ or $\frac{n_0}{k}-1$ nodes.
\end{proof}

\begin{lemma}\label{lemma:partition_time}
	Protocol \ref{protocol:partition}, terminates after $\Theta(kn^2)$ expected time.
\end{lemma}
\begin{proof}
	Protocol \ref{protocol:partition} operates in phases, where each phase doubles the number of partitions. After $\log{k}$ phases, there exist $k$ groups in the population and the nodes terminate.
	
	We now study the time that each group $c_i$ needs in order to split into two partitions. Here, for simplicity, $i$ indicates the level of a partition $c$ in the binary tree and $m_i$ the number of nodes of partition $c_i$.
	
	Let $X$ be a random variable defined to be the number of steps until all $m_i$ nodes move to their next partitions. Call a step a success if two nodes in $c_i$ interact, thus, moving to their next partitions. We divide the steps of the protocol into \textit{epochs}, where epoch $j$ begins with the step following the $j$th success and ends with the step at which the $(j+1)$st success occurs. Let also the r.v. $X_j$, $1 \leq j \leq m_i$ be the number of steps in the $j$th epoch.
		
	The probability of success during the $j$th epoch, for $0 \leq j \leq m_i$, is
	$p_j = \frac{(m_i-j)(m_i-j-1)}{n(n-1)}$ and $E[X_j] = 1/p_j$. By linearity of expectation we have

\begin{equation}
	\begin{split}
		E[X] &= E[\sum_{j=0}^{m_i-2}X_j] = \sum_{j=0}^{m_i-2}E[X_j] = n(n-1)\sum_{j=0}^{m_i-2}\frac{1}{(m_i-j)(m_i-j-1)} \\
		&= n(n-1)\sum_{j=2}^{m_i}\frac{1}{j(j-1)} < n(n-1) \sum_{j=2}^{m_i}\frac{1}{(j-1)^2} \\
		&= n(n-1) \sum_{j=1}^{m_i-1}\frac{1}{j^2} = n(n-1)(1-\frac{1}{m_i}) < n^2
	\end{split}
\end{equation}
	
\noindent The above uses the fact that $m_i \leq n$ for any $i \geq 0$.	
	
	For the lower bound, observe that the last two remaining nodes in $c_i$ need on average $n(n-1)/2$ steps to meet each other. Thus, we conclude that $E[X]= \Theta(n^2)$.
	
	In total, $\sum_{0}^{\log{(k)}-1}2^i = 2^{\log{k}}-1 = k-1$ partitions split, thus, the total expected time to termination is $\Theta(kn^2)$ steps.
\end{proof}

\begin{lemma}\label{lemma:partition_faults}
	In the case where up to $f$ faults occur during the execution of Protocol \ref{protocol:partition}, each final partition has at least $n/k - f$ nodes, where $k$ is the number of partitions and $f<k$.
\end{lemma}
\begin{proof}
	Consider the case where $f$ faults occur in the first partition $c_0$. Then, we can assume that we run a failure-free execution on a population of size $n_0=n-f$.
	By Lemma \ref{lemma:partition_correctness}, each partition will end up having either $\frac{n-f}{k} > \frac{n}{k}-2$ or $\frac{n-f}{k} - 1 > \frac{n}{k}-3$ nodes.
	
	Now, consider the case where $f$ crash faults occur in some partition $c_j$.
	The nodes of each partition $c_j$ operate independently from the rest of the population, that is, they never update their states and/or connections when they interact with nodes from a different partition. Thus, as in the previous case, if no more faults occur, we can assume that we have a failure-free execution on $|c_j|-f$ nodes.
	Again, by Lemma \ref{lemma:partition_correctness}, after $p$ subdivisions, each final partition that was obtained by $c_j$ will either have $\frac{|c_j|-f}{p}$ or $\frac{|c_j|-f}{p} - 1$ nodes.
	Consequently, any number of faults in a partition $c_i$ are equally split into the partitions following $c_i$.
	
	It is then obvious that in the worst case, a final partition might have $\frac{n}{k}-f-1$ nodes, and this is the result of $f$ faults in a final partition.
\end{proof}

By Lemma \ref{lemma:partition_faults}:
\begin{corollary}
	$||V_i| - |V_j|| \leq f+1$, $\forall 1 \leq i,j \leq k$.
\end{corollary}

By Lemma \ref{lemma:partition_faults} and the definition of partial constructibility (Definition \ref{def:constructible_language}):
\begin{theorem}
The language $L_{D,f}$, where $k$ is a constant number, is partially constructible under $f$ faults.
\end{theorem}

We now show that if we permit a waste linear in $n$, any graph language that is constructible in the fault-free NET model, becomes constructible under a bounded number of faults.

\begin{theorem}
	Take any NET protocol $\Pi$ of the original fault-free model.
	There is a NET $\Pi'$ such that when at most $f$ faults may occur on any population of size $n$, $\Pi'$ successfully simulates an execution of $\Pi$ on at least $\frac{n}{2f}-f$ nodes.
\end{theorem}
\begin{proof}
	Consider any constructible language $L$ and a protocol $\Pi$ that constructs it.
	For any bounded number of faults $f$, set $k=2^\delta$, where $2^{\delta-1}<f$.
	Consider a protocol $\Pi'$, which consists of the rules $1$ and $2$ of Protocol \ref{protocol:partition}.
	These rules partition the population into $k$ groups, where $k$ is an input parameter of $\Pi'$.
	By Lemma \ref{lemma:partition_faults}, each group has at least $n/k - f$ nodes. For $2f$ partitions in the worst case, the number of nodes in each partition is at least $\frac{n}{2f} - f$.
	Then assume that when a node reaches its final partition, it starts executing protocol $\Pi$, updating its state and connections only when interacting with nodes of the same partition.
	As the number of partitions is strictly more that the upper bound on the number of faults $f$, there exists at least one partition that no fault has occurred.
\end{proof}

\section{Notified Network Constructors}
\label{nNET}

In this section, in light of the impossibility result of Section \ref{fault_tolerant_SNET}, we allow fault notifications when nodes crash. In particular, we introduce a \textit{fault flag} in each node, which is initially zero.
When a node $u$ crashes at time $t$, every node $v$ which was adjacent to $u$ at time $t$ is notified, that is, the fault flag of all $v$ becomes $1$. In the case where $u$ is an isolated node (i.e., it has no active edges), an arbitrary node $w$ in the graph is notified, and its fault flag becomes $2$. Then, the fault flag becomes immediately zero after applying a corresponding rule from the transition function.

More formally, the set of node-states is $Q \times \{0, 1, 2\}$, and for clarity in our descriptions and protocols, we define two types of transition functions. The first one determines the node and connection state updates of pairwise interactions ($\delta_1: Q \times Q \times \{0,1\} \rightarrow Q \times Q \times \{0,1\}$), while the second transition function determines the node state updates due to fault notifications ($\delta_2: Q \times \{0,1,2\} \rightarrow Q \times \{0,1,2\} $). This means that during a step $t$ that a node $u$ crashes, all its adjacent nodes are allowed to update their states based on $\delta_2$ at that same step. If there are no any adjacent nodes to $u$, an arbitrary node is notified, thus updating its state based on $\delta_2$ at step $t$.

We have assumed that the faults can only occur sequentially (at most one fault per step). This assumption was equivalent to the case where many faults can occur in each step in the original NET model. However, when fault notifications are allowed, this does not hold, unless the fault flag could be used as a counter of faults in each step. We want to keep the model as minimal as possible, thus, we only allow the adversary to choose one node at most in each step to crash.

As long as only one fault at most can occur in each step, the separation of these transition functions is equivalent to the case where only one transition function exists $\delta: (Q \times\{0,1,2\}) \times (Q\times\{0,1,2\}) \times \{0,1\} \rightarrow (Q\times\{0,1,2\}) \times (Q\times\{0,1,2\}) \times \{0,1\}$.
Consider the case where a node $u$ crashes, notifying a node $w$ in the population (its fault flag becomes either $1$ or $2$). Then, in the first case (separate transition functions), $w$ is instantly allowed to update its state, while in the second case (unified transition functions), $w$ waits until its next interaction with a node $v$, applying the rule of $\delta_2$ independently of the state and connection of $v$. During the same interaction, $w$ and $v$ can also update their states and connections based on the corresponding rule of $\delta_1$.

In this section, we investigate whether the additional information in each agent (the fault flag) is sufficient in order to design fault-tolerant or $k-$fault-tolerant protocols, overcoming the impossibility of certain graph languages in the NET model.

Such a minimal fault notification mechanism can be exploited to construct a larger class of graph languages that in the original Network Constructors model where no form of notifications was available.

\subsection{Fault-Tolerant Protocols}\label{nNET_minimal_updates}

In this section, our goal is to design protocols that after a fault, the nodes try to fix the configuration and eventually stabilize to a correct network.
We give protocols for some basic network construction problems, such as \textit{spanning star}, \textit{cycle cover}, and in Section \ref{universal_waste} we give a fault-tolerant spanning line protocol which is part of our generic constructor capable of constructing a large class of networks.

\begin{algorithm}[H]
	\floatname{algorithm}{Protocol} 
	\caption{FT Spanning Star}\label{protocol:spanning_star}
	\begin{algorithmic}[1000]
		
		\State $Q = \{b,\; r\}\times\{0,1\}$
		\State Initial state: $b$
		\State $ $
		\State $\delta_1:$

		\State $(b,\; b, \; 0) \rightarrow (b,\; r, \; 1)$		
		\State $(b,\; b, \; 1) \rightarrow (b,\; r, \; 1)$		
		\State $(r,\; r, \; 1) \rightarrow (b,\; b, \; 0)$		
		\State $(b,\; r, \; 0) \rightarrow (b,\; r, \; 1)$
		
		\State $ $
		
		\State $\delta_2:$
		\State $(r,1) \rightarrow (b,0)$

	\end{algorithmic}
\end{algorithm}

\begin{proposition}\label{lemma:spanning_star1}
	\textit{FT Spanning Star} is fault-tolerant.
\end{proposition}
\begin{proof}
	Assume that any number of faults $k<n$ occur during an execution. Initially, all nodes are in state $b$ (\textit{black}).
	Two nodes connect with each other, if either one of them is black, or both of them are black,	in which case one of them becomes $r$ (\textit{red}).
	A black node can become red only by interaction with another black node, in which case they also become connected.
	Thus, with no crash faults, a connected component always includes at least one black node. In addition, all isolated nodes are always in state $b$. This is because, if a red node removes an edge it becomes black.
	
	Then, if a (connected) node crashes, the adjacent nodes are notified and the red nodes become black, thus, any connected component should again include at least one black node.
	Now, consider the case where only one black node remains in the population. Then the rest of the population (in state $r$) should be in the same connected component as the unique $b$ node.
	Then, if $b$ crashes, at least one black node will appear, thus, this protocol maintains the invariant, as there is always at least one black node in the population.
	\textit{FT Spanning Star} then stabilizes to a star with a unique black node in the center.
\end{proof}

\begin{algorithm}[H]
	\floatname{algorithm}{Protocol} 
	\caption{FT Cycle-Cover}\label{protocol:cycle_cover}
	\begin{algorithmic}[1000]
		
		\State $Q = \{q_0,\; q_1,\; q_2\} \times \{0,1\}$
		\State Initial state: $q_0$
		\State $ $
		
		\State $\delta_1:$
		
		\State $(q_0, \; q_0, \; 0) \rightarrow (q_1, \; q_1, \; 1)$		
		\State $(q_1, \; q_0, \; 0) \rightarrow (q_2, \; q_1, \; 1)$		
		\State $(q_1, \; q_1, \; 0) \rightarrow (q_2, \; q_2, \; 1)$
		
		\State $ $
		
		\State $\delta_2:$
		\State $(q_1,1) \rightarrow (q_0,0)$
		\State $(q_2,1) \rightarrow (q_1,0)$			
		
	\end{algorithmic}
\end{algorithm}

\noindent Similarly, we can show the following.

\begin{proposition}
	\textit{FT Cycle-Cover} is fault-tolerant.
\end{proposition}

\subsection{Universal Fault-Tolerant Constructors}\label{universal_waste}

In this section, we ask whether there is a generic fault-tolerant constructor capable of constructing a large class of graphs.
We first give a fault-tolerant protocol that constructs a spanning line, and then we show that we can simulate a given TM on that line, tolerating any number of crash faults.
Finally, we exploit that in order to construct any graph language that can be decided by an $O(n^2)-$space TM, paying at most linear waste.

\begin{lemma}\label{lemma:spanning_line1}
	\textit{FT Spanning Line} (Protocol \ref{protocol:ft_spanning_line}) is fault-tolerant.
\end{lemma}
\begin{proof}
	Initially, all nodes are in state $q_0$ and they start connecting with each other in order to form lines that eventually merge into one.
	
	When two $q_0$ nodes become connected, one of them becomes a leader (state $l_0$) and starts connecting with $q_0$ nodes (expands). A leader state $l_0$ is always an endpoint. The other endpoint is in state $e_i$ (initially $e_1$), while the inner nodes are in state $q_2$. Our goal is to have only one leader $l_0$ on one endpoint, because $l_0$ are also used in order to merge lines. Otherwise, if there are two $l_0$ endpoints, the line could form a cycle.
	
	When two $l_0$ leaders meet, they connect (line merge) and a $w$ node appears. This state performs a random walk on the line and its purpose is to meet both endpoints (at least once) before becoming an $l_0$ leader. After interacting with the first endpoint, it becomes $w_1$ and changes the endpoint to $e_1$. Whenever it interacts with the same endpoint they just swap their states from $e_1$, $w_1$ to $e_2$, $w_2$ and vice versa. In this way, we guarantee that $w_i$ will eventually meet the other endpoint in state $e_j,\; j \neq i$, or $l_0$. In the first case, the $w_i$ node becomes a leader ($l_0$), after having walked the whole line at least once.
	
	Now, consider the case where a fault may happen on some node on the line. If the fault flag of an endpoint state becomes $1$, it updates its state to $q_0$.
	Otherwise, the line splits into two disjoint lines and the new endpoints become $l_1$. An $l_1$ becomes a walking state $w_1$, changes the endpoint into $e_1$ and performs the same process (random walk).
	
	If there are more than one walking states on a line, then all of them are $w$, or $w_i$ and they perform a random walk. None of them can ever satisfy the criterion to become $l_0$ before first eliminating all the other walking states and/or the unique leader $l_0$ (when two walking states meet, only one survives and becomes $w$), simply because they form natural obstacles between itself and the other endpoint. If a new fault occurs, then this can only introduce another $w_i$ state which cannot interfere with what existing $w_i$'s are doing on the rest of the line (can meet them eventually but cannot lead them into an incorrect decision).

	If an $l_0$ leader is merging while there are $w_i$'s and/or $w$'s on its line (without being aware of that), the merging results in a new $w$ state, which is safe because a $w$ cannot make any further progress without first succeeding to beat everybody on the line. A $w$ can become $l_0$ only after walking the whole line at least once (i.e., interact with both endpoints) and to do that it must have managed to eliminate all other walking states of the line on its way.

	We have shown that despite the presence of faults, any expansion or merging eventually succeeds, meaning that the population eventually forms a line with a single leader in one endpoint.
\end{proof}

\begin{algorithm}[H]
	\floatname{algorithm}{Protocol} 
	\caption{FT Spanning Line}\label{protocol:ft_spanning_line}
	\begin{algorithmic}[1000]
		
		\State $Q = \{q_0,\; q_2,\; e_1,\; e_2,\; l_0,\; l_1,\; w,\; w_1,\; w_2 \} \times \{0,1\}$
		\State Initial state: $q_0$
		\State $ $
		
		\State $\delta_1:$
		
		\State $(q_0, \; q_0, \; 0) \rightarrow (e_1, \; l_0, \; 1)$
		\State $(l, \; q_0, \; 0) \rightarrow (q_2, \; l_0, \; 1)$
		\State $(l_0, \; l_0, \; 0) \rightarrow (q_2, \; w, \; 1)$

		\State $ $
		
		\State \textbackslash \textbackslash $w$ nodes perform a random walk on line
		\State $(l_1, \; q_2, \; 1) \rightarrow (e_1, \; w_1, \; 1)$
		\State $(w_i, \; q_2, \; 1) \rightarrow (q_2, \; w_i, \; 1)$
		\State $(w, \; q_2, \; 1) \rightarrow (q_2, \; w, \; 1)$
		\State $ $

		\State $(w, \; e_i, \; 1) \rightarrow (w_i, \; e_i, \; 1)$
		\State $(w_i, \; e_i, \; 1) \rightarrow (w_j, \; e_j, \; 1)$, $i \neq j$
		\State $(w_i, \; e_j, \; 1) \rightarrow (q_2, \; l_0, \; 1)$, $i \neq j$
		\State $(w, \; l_i, \; 1) \rightarrow (w_1, \; e_1, \; 1)$
		\State $(w_i, \; l_i, \; 1) \rightarrow (q_2, \; l_0, \; 1)$
		\State $ $
		
		\State $ $
		
		\State \textbackslash \textbackslash $w$ nodes eliminate each other, until only one survives
		\State $(w_i, \; w_j, \; 1) \rightarrow (w, \; q_2, \; 1)$
		\State $(w, \; w_j, \; 1) \rightarrow (w, \; q_2, \; 1)$

		\State $ $
		
		\State $\delta_2:$
		\State $(e_1,1) \rightarrow (q_0,0)$	
		\State $(e_2,1) \rightarrow (q_0,0)$	
		\State $(l_0,1) \rightarrow (q_0,0)$	
		\State $(l_1,1) \rightarrow (q_0,0)$
		\State $(q_2,1) \rightarrow (l_1,0)$	
		\State $(w,1) \rightarrow (l_1,0)$	
		\State $(w_1,1) \rightarrow (l_1,0)$	
		\State $(w_2,1) \rightarrow (l_1,0)$

	\end{algorithmic}
\end{algorithm}


\begin{lemma}\label{lemma:spanning_line_tm}
	There is a NET $\Pi$ (with notifications) such that when $\Pi$ is executed on $n$ nodes and at most $k$ faults can occur, $0 \leq k < n$, $\Pi$ will eventually simulate a given TM $M$ of space $O(n-k)$ in a fault-tolerant way.
\end{lemma}
\begin{proof}	
	
	The state of $\Pi$ has two components $(P,S)$, where $P$ is executing a spanning line formation procedure, while $S$ handles the simulation of the TM $M$.
	Our goal is to eventually construct a spanning line, where initially the state of the second component of each node is in an initial state $s_0$ except from one node which is in state \textit{head} and indicates the head of the TM.

	In general, the states $P$ and $S$ are updated in parallel and independently from each other, apart from some cases where we may need to reset either $P$, $S$ or both.
	
	In order to form a spanning line under crash failures, the $P$ component will be executing our \textit{FT Spanning Line} protocol which is guaranteed to construct a line, spanning eventually the non-faulty nodes.
	
	It is sufficient to show that the protocol can successfully reinitialize the state of all nodes on the line after a final \textit{event} has happened and the line is stable and spanning. Such an \textit{event} can be a line merging, a line expansion, a fault on an endpoint or an intermediate fault. The latter though can only be a final event if one of the two resulting lines is completely eliminated due to faults before merging again. In order to re-initialize the TM when the line expands to an isolated node $q_0$, we alter a rule of the \textit{FT Spanning Line} protocol. Whenever, a leader $l_0$ expands to an isolated node $q_0$, the leader becomes $q_2$ while the node in $q_0$ becomes $l_1$, thus introducing a new walking state.
	
	We now exploit the fact that in all these cases, \textit{FT Spanning Line} will generate a $w$ or a $w_i$ state in each affected component.
	
	Whenever a $w_1$ or $w_2$ state has just appeared or interacted with an endpoint $e_1$ or $e_2$ respectively, it starts resetting the simulation component $S$ of every node that it encounters. If it ever manages to become a leader $l_0$, then it finally restarts the simulation on the $S$ component by reintroducing to it the \textit{tape head}.
	
	When the last event occurs, the final spanning line has a $w$ or $w_i$ leader in it, and we can guarantee a successful restart due to the following invariant.
	Whenever a line has at least one $w/w_i$ state and no further events can happen, \textit{FT Spanning Line} guarantees that there is one $w$ or $w_i$ that will dominate every other $w/w_i$ state on the line and become an $l_0$, while having traversed the line from endpoint to endpoint at least once.
	
	In its final departure from one endpoint to the other, it will dominate all $w$ and $w_i$ states that it will encounter (if any) and reach the other endpoint. Therefore, no other $w/w_i$ states can affect the simulation components that it has reset on its way, and upon reaching the other endpoint it will successfully introduce a new \textit{head} of the TM while all simulation components are in an initial state $s_0$.
\end{proof}

\begin{lemma}\label{lemma:partition}
	There is a fault-tolerant NET protocol $\Pi$ (with notifications) which partitions the nodes into two groups $U$ and $D$ with waste at most $2f(n)$, where $f(n)$ is an upper bound on the number of faults that can occur. $U$ is a spanning line with a unique leader in one endpoint and can eventually simulate a TM $M$. In addition, each node of $D$ is connected with exactly one node of $U$, and vice versa.
\end{lemma}
\begin{proof}
	Initially all nodes are in state $q_0$. Protocol $\Pi$ partitions the nodes into two equal sets $U$ and $D$ and every node maintains its type forever. This is done by a perfect matching between $q_0$'s where one becomes $q_u$ and the other becomes $q_d$.
	Then, the nodes of $U$ execute the \textit{FT Spanning Line} protocol, which guarantees the construction of a spanning line, capable of simulating a TM (Lemma \ref{lemma:spanning_line_tm}).
	The rest of the nodes ($D$), which are connected to exactly one node of $U$ each, are used to construct on them random graphs.	
	Whenever a line merges with another line or expands towards an isolated node, the simulation component $S$ in the states of the line nodes, as described in Lemma \ref{lemma:spanning_line_tm}, is reinitialised sequentially.

	Assume that a fault occurs on some node of the perfect matching before that pair has been attached to a line. In this case, it's pair will become isolated therefore it is sufficient to switch that back to $q_0$.
	
	If a fault occurs on a $D$ node $u$ after its pair $w$ has been attached to a line, $w$ goes into a detaching state which disconnects it from its line neighbors, turning them into $l_1$ and itself becoming a $q_0$ upon release. An $l_1$ state on one endpoint is guaranteed to walk the whole line at least once (as $w_i$) in order to ensure that a unique leader $l_0$ will be created.
	  If $u$ fails before completing this process, it's neighbors on the line shall be notified becoming again $l_1$, and if one of its neighbors fails we shall treat this as part of the next type of faults. This procedure shall disconnect the line but may leave the component connected through active connections within $D$. But this is fine as long as the \textit{FT-Spanning Line} guarantees a correct restart of the simulation after any event on a line. This is because eventually the line in $U$ will be spanning and the last event will cause a final restart of the simulation on that line.

	Assume that a fault occurs on a node $u \in U$ that is part of the line. In this case the neighbors of $u$ on the line shall instantly become $l_1$.
	Now, its $D$ pair $v$, which may have an unbounded number of $D$ neighbors at that point, becomes a special \textit{deactivating state} that eventually deactivates all connections and never participates again in the protocol, thus, its stays forever as waste. This is because the fault partially destroys the data of the simulation, thus, we cannot safely assume that we can retrieve the degree of $v$ and successfully deactivate all edges.
	As there can be at most $f(n)$ such faults we have an additional waste of $f(n)$.
	Now, consider the case where $u$ is one neighbor of a node $w$ which is trying to release itself after its $v$ neighbor in $D$ failed.
	Then, $w$ implements a $2$-counter in order to remember how many of its alive neighbours have been deactivated by itself or due to faults in order to know when it should become $q_0$.
\end{proof}

\begin{theorem}\label{theorem:linear_space_tm}
	For any graph language $L$ that can be decided by a linear space TM, there is a fault-tolerant NET $\Pi$ (with notifications) that constructs a graph in $L$ with waste at most $min\{n/2 + f(n), \; n\}$, where $f(n)$ is an upper bound on the number of faults that can occur.
\end{theorem}
\begin{proof}
	By Lemma \ref{lemma:partition}, there is a protocol that constructs two groups $U$ and $D$ of equal size, where each node of $U$ is matched with exactly one node of $D$, and vice versa. In addition, the nodes of $U$ form a spanning line, and by Lemma \ref{lemma:spanning_line_tm} it can simulate a TM $M$. After the last fault occurs, $M$ is correctly initialized and the head of the TM is on one of the endpoints of the line.
	The two endpoints are in different states, and assume, that the endpoint that the head ends up is in state $e_l$ (\textit{left} endpoint), and the other is in state $e_r$ (\textit{right} endpoint).
	
	We now provide the protocol that performs the simulation of the TM $M$, which we separate into several subroutines.
	The first subroutine is responsible for simulating the direction on the tape and is executed once the head reaches the endpoint $e_l$.
	The simulation component $S$ (as in Lemma \ref{lemma:spanning_line_tm}) of each node has three sub-components $(h,c,d)$. $h$ is used to store the head of the TM, i.e., the actual state of the control of the TM, $c$ is used to store the symbol written on each cell of the TM, and $d$ is either $l$, $r$ or $\sqcup$, indicating whether that node is on the left or on the right of the head (or unknown). Assume that after the initialization of the TM, $d=\sqcup$ for all nodes of the line. Finally, whenever the head of the TM needs to move from a node $u$ to a node $w$, $h_w \leftarrow h_u$, and $h_u \leftarrow \sqcup$.
	
	\textit{Direction}. Once the head of the TM is introduced in the endpoint $e_l$ by the lines' leader, it moves on the line, leaving $l$ marks on the $d$ component of each node. It moves on the nodes which are not marked, until it eventually reaches the $e_r$ endpoint. At that point, it starts moving on the marked nodes, leaving $r$ marks on its way back. Eventually, it reaches again the $e_l$ endpoint. At that time, for each node on its right it holds that $d=r$.
	Now, every time it wants to move to the right it moves onto the neighbor that is marked by $r$ while leaving an $l$ mark on its previous position, and vice versa.
Once the head completes this procedure, it is ready to begin working as a TM.

	\textit{Constructing a random graph in $D$}. This subroutine of the protocol constructs a random graph in the nodes of $D$. Here, the nodes are allowed to toss a fair coin during an interaction. This means that we allow transitions that with probability $1/2$ give one outcome and with $1/2$ another. To achieve the construction of a random graph, the TM implements a binary counter $C$ ($\log{n}$ bits) in its memory and uses it in order to uniquely identify the nodes of set $D$ according to their distance from $e_l$. Whenever it wants to modify the state of edge $(i,j)$ of the network in $D$, the head assigns special marks to the nodes in $D$ at distances $i$ and $j$ from the left of the endpoint $e_l$. Note that the TM uses its (distributed) binary counter in order to count these distances. If the TM wants to access the $i-$th node in $D$, it sets the counter $C$ to $i$, places a mark on the left endpoint $e_l$ and repeatedly moves the mark one position to the right, decreasing the counter by one in each step, until $C=0$. Then, the mark has been moved exactly $i$ positions to the right.
	In order to construct a random graph in $D$, it first assigns a mark $r_1$ to the first node $e_l$, which indicates that this node should perform random coin tosses in its next interactions with the other marked nodes, in order to decide whether to form connections with them, or not. Then, the leader moves to the next node on its line and waits to interact with the connected node in $D$. It assigns a mark $r_2$, and waits until this mark is deleted.
	The two nodes that have been marked ($r_1$ and $r_2$), will eventually interact with each other, and they will perform the (random) experiment. Finally the second node deletes its mark ($r_2$).
	The head then, moves to the next node and it performs the same procedure, until it reaches the other endpoint $e_r$. Finally, it moves back to the first node (marked as $r_1$), deletes the mark and moves one step right.
	This procedure is repeated until the node that should be marked as $r_1$ is the right endpoint $e_r$. It does not mark it and it moves back to $e_l$.
	The result is an equiprobable construction of a random graph. In particular, all possible graphs over $|D|$ nodes have the same probability to occur.
	Now, the input to the TM $M$ is the random graph that has been drawn on $D$, which provides an encoding equivalent to an adjacency matrix. Once this procedure is completed, the protocol starts the simulation of the TM $M$.
	There are $m=(\frac{k}{2})^2$ edges, where $k=|D|$ and $M$ has available $\frac{k}{2}=\sqrt{m}$ space, which is sufficient for the simulation on a $\sqrt{m}-$space TM.
	
	\textit{Read edges of $D$}. We now present a mechanism, which can be used by the TM in order to read the state of an edge joining two nodes in $D$. Note that a node in $D$ can be uniquely identified by its distance from the endpoint $e_l$. Whenever the TM needs to read the edge joining the nodes $i$ and $j$, it sets the counter $C$ to $i$. Assume w.l.o.g. that $i<j$. It performs the same procedure as described in the subroutine which draws the random graph in $D$. It moves a special mark to the right, decreasing $C$ by one in each step, until it becomes zero. Then, it assigns a mark $r_3$ on the $i-$th node of $D$, and then performs the same for $C=j$, where it also assigns a mark $r_4$ (to the $j-$th node).
	When the two marked nodes ($r_3$ and $r_4$) interact with each other, the node which is marked as $r_4$ copies the state of the edge joining them to a flag $f$ (either $0$ or $1$), and they both delete their marks.
	The head waits until it interacts again with the second node, and if the mark has been deleted, it reads the value of the flag $f$.
	
	
	After a simulation, the TM either accepts or rejects. In the first case, the constructed graph belongs to $L$ and the Turing Machine halts. Otherwise, the random graph does not belong to $L$, thus the protocol repeats the random experiment. It constructs again a random graph, and starts over the simulation on the new input.

	A final point that we should make clear is that if during the simulation of the TM an event occurs (crash fault, line expansion, or line merging), by Lemma \ref{lemma:spanning_line_tm} and Lemma \ref{lemma:partition}, the protocol reconstructs a valid partition between $U$ and $D$, the TM is re-initialized correctly, and a unique head is introduced in one endpoint. At that time, edges in $D$ may exist, but this fact does not interfere with the (new) simulation of the TM, as a new random experiment takes place for each pair of nodes in $D$ prior to each simulation.
\end{proof}


We now show that if the constructed network is required to occupy $1/3$ instead of half of the nodes, then the available space of the TM-constructor dramatically increases from $O(n)$ to $O(n^2)$.
We provide a protocol which partitions the population into three sets $U$, $D$ and $M$ of equal size $k=n/3$. The idea is to use the set $M$ as a $\Theta(n^2)$ binary memory for the TM, where the information is stored in the $k(k-1)/2$ edges of $M$.

\begin{algorithm}[H]
	\floatname{algorithm}{Protocol} 
	\caption{3-Partition}\label{protocol:three_partition}
	\begin{algorithmic}[1000]
		
		\State $Q = \{q_0,\; q_d,\; q_u ,\; q_u',\; q_m ,\; q_m',\; q_w,\; q_w' ,\; s\} \times \{0,1\}$
		\State Initial state: $q_0$
		\State $ $
		
		\State $\delta_1:$
		
		\State $(q_0, \; q_0, \; 0) \rightarrow (q_u', \; q_d, \; 1)$		
		\State $(q_u', \; q_0, \; 0) \rightarrow (q_u, \; q_m, \; 1)$		
		\State $(q_u', \; q_u', \; 0) \rightarrow (q_u, \; q_m', \; 1)$
		\State $(q_m', \; q_d, \; 1) \rightarrow (q_m, \; q_0, \; 0)$
		\State $(q_w, \; q_d, \; 1) \rightarrow (q_0, \; s, \; 0)$
		\State $(q_w, \; q_u, \; 1) \rightarrow (q_m, \; q_u, \; 1)$	
		\State $(q_w', \; q_d, \; 1) \rightarrow (q_0', \; s, \; 0)$
		\State $(q_w', \; q_m, \; 1) \rightarrow (q_0', \; s, \; 0)$
		\State $(q_w', \; q_m', \; 1) \rightarrow (q_0', \; q_u', \; 0)$
				
		\State $(s, \cdot, \; 1) \rightarrow (s, \cdot, \; 0)$	
		
		\State $ $
		
		\State $\delta_2:$
		\State $(q_u',1) \rightarrow (q_0,0)$
		\State $(q_d,1) \rightarrow (s,0)$		
		\State $(q_m,1) \rightarrow (s,0)$	
		\State $(q_w,1) \rightarrow (q_0,0)$
		\State $(q_w',1) \rightarrow (q_0',0)$
		
		\State $(q_m',1) \rightarrow (q_w,0)$		
		\State $(q_u,1) \rightarrow (q_w',0)$

	\end{algorithmic}
\end{algorithm}

\begin{lemma}\label{three_partition}
	Protocol 3-Partition partitions the nodes into three groups $U$, $D$ and $M$, with waste $3f(n)$, where $f(n)$ is an upper bound on the number of faults that can occur. $U$ is a spanning line with a unique leader in one endpoint and can eventually simulate a TM, each node in $D \cup M$ is connected with exactly one node of $U$, and each node of $U$ is connected to exactly one node in $D$ and one node in $M$.
\end{lemma}
\begin{proof}
	Protocol $3-$Partition constructs lines of three nodes each, where one endpoint is in state $q_d$, the other endpoint in state $q_m$, and the center is in state $q_u$. The nodes of $U$ operate as in Lemma \ref{lemma:partition} (i.e., they execute the \textit{FT Spanning Line} protocol).
	A (connected) pair of nodes waits until a third node is attached to it, and then the center becomes $q_u$ and starts executing the \textit{FT Spanning Line} protocol.
	Note that at some point, it is possible that the population may only consists of pairs in states $q_d$ and $q_u'$. For this reason, we allow $q_u'$ nodes to connect with each other, forming lines of four nodes. One of the $q_u'$ nodes becomes $q_u$ and the other becomes $q_m'$. A node in $q_m'$ becomes $q_m$ only after deactivating its connection with a $q_d$ node (its previous pair). This results in lines of three nodes each with nodes in states $q_d$, $q_u$ and $q_m$. Then, the $q_u$ nodes start forming a line, spanning all nodes of $U$.
	In a failure-free setting, the correctness of this protocol follows from Lemma \ref{lemma:partition}. In addition, by Lemma \ref{lemma:spanning_line_tm}, the TM of the line is initialized correctly after the last occurring event (line expansion, line merging, or crash fault).
	
	If we consider crash failures, it is sufficient to show that eventually $U$ is a spanning line and $M$ and $D$ are disjoint.
	If a node ever becomes $q_d$ or $q_m$, it might form connections with other nodes in $D$ or $M$ respectively, because of a TM simulation. A node in $M$ never forms connections with nodes in $D$.
	After they receive a fault notification, they become the \textit{deactivating state} $s$. A node in state $s$ is disconnected from any other node, thus, it eventually becomes isolated and never participates in the execution again. We do this because nodes in $M$ and $D$ can form unbounded number of connections. The data of the TM have been partially destroyed (because of the crash failure), therefore it is not safe to assume that we can retrieve the degree of them and successfully re-initialize them.
	
	A node $u$ in state $q_m'$ (inner node of a line of four nodes), after a fault notification it becomes $q_w$. A node in $q_w$ waits until its next interaction with a connected node $v$. If $v$ is in state $q_u$, this means that now a triple has been formed, thus $u$ becomes $q_m$. If $v$ is in state $q_d$, they delete the edge joining them, $u$ becomes $q_0$ and $v$ becomes $s$ ($v$ might have formed connections with other nodes in $D$).
	
	A node $u$ in $q_u$, after a fault notification it becomes $q_w'$ and waits until its next interaction with a connected node $v$. At that point, $v$ can be either $q_d$, $q_m'$, or $q_m$. In all cases they disconnect from each other and $u$ becomes $q_0'$. The state $q_0'$ indicates that the node should release itself from the spanning line in $U$. This procedure works as described in Lemma \ref{lemma:partition}, thus, after releasing itself from the line, it becomes $q_0$. If $v$ is in state $q_d$ or $q_m$, it becomes $s$. If $v$ is in state $q_m'$, it becomes $q_u'$, as its (unique) adjacent node can only be in state $q_d$.

	A node in $q_u'$ or $q_w$, after a fault notification it becomes $q_0$ and continues participating in the execution again.
	Finally, a node in state $q_w'$, after receiving a fault notification, it becomes $q_0'$ (a $q_w'$ is the result of a fault notification in a $U-$ node).
	
	Note that a node in any state except from $q_d$ and $q_m$ can be re-initialized correctly, thus they may participate in the execution again.
	It is apparent that no node that might have formed unbounded number of connections can participate in the execution again after a crash fault. This guarantees that the connections in $D$ and $M$ can be correctly initialized after the final event, and that no node in $D \cup M$ can be connected with more than one node in $U$.
	In addition, if a $U-$node receives a fault notification, it releases itself from the line, thus introducing new walking states in the resulting line(s).
	By Lemma \ref{lemma:spanning_line_tm}, this guarantees the correct re-initialization of the TM.
	Finally, a crash failure can lead in deactivating two more nodes, in the worst case. These nodes never participate in the execution again, thus they remain forever as waste. This means that after $f(n)$ crash failures, the partitioning will be constructed in $n - 3f(n)$ nodes.
\end{proof}

\begin{theorem}
	For any graph language $L$ that can be decided by an $O(n^2)-$space TM, there is a protocol that constructs $L$ equiprobably with waste at most $min\{2n/3 + f(n), \; n\}$, where $f(n)$ is an upper bound on the number of faults.
\end{theorem}
\begin{proof}
	Protocol \ref{protocol:three_partition} partitions the population in three groups $U$, $D$ and $M$ and by Lemma \ref{three_partition}, it tolerates any number of crash failures, while initializing correctly the TM after the final event (line expansion, line merging, or crash fault).
	Reading and writing on the edges of $M$ is performed in precisely the same way as reading/writing the edges of $D$ (described in Theorem \ref{theorem:linear_space_tm}).
	Thus, the Turing Machine has now a $O(n^2)-$space binary memory (the edges of $M$) and $O(n)-$space on the edges of the spanning line $U$. The random graph is constructed on the $k$ nodes of $D$ (useful space), where by Lemma \ref{three_partition}, $k=(n-3f(n))/3=n/3 - f(n)$ in the worst case.
\end{proof}

\subsection{Designing Fault-Tolerant Protocols without Waste}\label{nNET_restart}

A very simple, (yet impractical) idea that could tolerate any number $k < n$ of faults is to restart the protocol each time a node crashes.
The implementation of this idea requires the ability of some nodes to detect the removal of a node.

\begin{definition}
	Consider any execution $E_i$ of a finite protocol $\Pi$.
	There exists a finite number of different executions, and for each execution a step $t_i$ that $\Pi$ stabilizes.
	Call $C_{i,j}$ the $j-$th configuration of execution $E_i$, where $j \leq t_i$. Then, we call \textit{maximum reachable degree} of $\Pi$ the value $\textit{d} = max\{\text{Degree}(G(C_{i,j}))\}, \; \forall i,j$.
\end{definition}

We first show that even in the case where the whole population is notified about a crash failure, global restart is \textit{impossible for protocols with unbounded maximum reachable degree, if the nodes have constant memory}.
However, we provide a protocol that restarts the population, but we supply the nodes with $O(\log{n})$ bits of memory. In our approach, we use fault notifications, and if a node $w$ crashes, the set $N_w$ of the nodes that are notified, has the task to restart the protocol (i.e., to convert the current configuration into an initial one).

Consider a protocol $\Pi$ with the initial state $q_0$. We define as global restart the process which leads all alive nodes to the initial state $q_0$ without any enabled connections among them and then $\Pi$ gradually starts again. \\


\begin{theorem}\label{theorem:global_restart}
	Consider a protocol $\Pi$ with unbounded maximum reachable degree. Then, global restart of $\Pi$ is impossible for nodes with constant memory, even if every node $u$ in the population is notified about the crash failure.
\end{theorem}
\begin{proof}
	Consider a protocol $\Pi$ with constant number of states $k$ and unbounded \textit{maximum reachable degree}, which stabilizes to a graph $G$ of type $L$.
	Then any degree more than $k$ cannot be remembered by a node, that is, a state $q$ cannot indicate the degree of a node.
	
	Assume that at time $t$ a crash failure occurs and that there are some edges in the graph (call them \textit{spurious edges}).

	Protocol $\Pi$ is allowed to have rules that are triggered by the fault and try to erase those edges (\textit{erasing process}). We assume that all nodes in the population are notified about the crash failure.
	But, as long as the nodes are not aware of their degree, they do not know when the edge erasing process stops in order to allow the restart.
	To stop the erasing process is equivalent to counting the remaining edges and wait until the degree reaches zero.
	After a node deletes an edge it either stays in the same state or updates it in order to remember it. No more than $k$ such changes can happen, thus it is impossible to delete all edges and restart $\Pi$ with constant memory.
	
	So, any self-stabilizing protocol will inherit (after restarting gradually) some arbitrary spurious edges. 
	Thus, global restart is impossible.	
\end{proof}

A very interesting related question is to ask whether a protocol $\Pi$ with unbounded \textit{maximum reachable degree} can still stabilize to a correct graph after an unsuccessful restart, where some edges exist in the beginning of the execution. This is equivalent to ask whether $\Pi$ can still stabilize to a correct $G$, is we enable arbitrarily some connections prior to the execution.

\begin{theorem}\label{theorem:global_restart2}
	Consider a NET protocol $\Pi$ which stabilizes to a graph $G \in L$. Given that all nodes are in an initial state $q_0$ and assuming an adversary that can initialize arbitrarily any subset of edges among nodes, $\Pi$ stabilizes to a graph $G' \notin L$.
\end{theorem}
\begin{proof}
	Assume w.l.o.g. that $\Pi$ stabilizes to a spanning line.
	Since the nodes have constant memory (i.e., constant number of states), there exists at least one state $q_1$ which $O(n)$ nodes stabilize to.
	Consider an execution $E$ where two nodes $v$ and $w$ are in the same state $q_1$ after stabilization at time $t$. Consider also a node $u$ in state $q_2$ which is adjacent to $v$ but not to $w$, and that $u$ and $w$ never interacted with each other until time $t$.
	
	Consider now that the adversary initializes the edge between $u$ and $w$ to \textit{on}, and we run an execution of $\Pi$ which is exactly the same as $E$ ($u$ and $w$ won't update their connection state, as they do not interact until $t'>t$).
	Then, node $u$ stabilizes having three enabled connections. Since $v$ and $w$ are both in the same state $q_1$, $u$ cannot distinguish $v$ and $w$. If there was a rule in $\Pi$ which disconnects $q_2$ and $q_1$, this would also happen in the case where $u$ was not adjacent to $w$, resulting $\Pi$ to stabilize to a graph with at least two disjoint lines, as $u$ would be disconnected from $v$.
\end{proof}

In light of the impossibility result of Theorem \ref{theorem:global_restart}, we allow the nodes to use non-constant local memory in order to develop a fault tolerating procedure based on restart.
Our goal is to come up with a protocol $A$ that can be composed with any NET protocol $\Pi$ (with notifications), so that their composition is a fault-tolerant version of $\Pi$.
Essentially, whenever a fault occurs, $A$ will restart all nodes in a way equivalent to as if a new execution of $\Pi$ had started on the whole remaining population.

We give a protocol that achieves this as follows.
All nodes are initially leaders. Through a standard pairwise leader elimination procedure, a unique leader would be guaranteed to remain in the absence of failures.
But because a fault can remove the last remaining leader, the protocol handles this by generating a new leader upon getting a fault notification.
This guarantees the existence of at least one leader in the population and eventually (after the last fault) of a unique one.
There are two main events that trigger a new restarting phase: a fault and a leader elimination.
As any new event must trigger a new restarting phase that will not interfere with an outdated one, eventually overriding the latter and restarting all nodes once more, we use phase counters to distinguish among phases.
In the presence of a new event it is always guaranteed that a leader at maximum phase will eventually increase its phase, therefore a restart is guaranteed after any event. The restarts essentially cause gradual deactivation of edges (by having nodes remember their degree throughout) and restoration of nodes' states to $q_0$, thus executing $\Pi$ on a fresh initial configuration.
For the sake of clarity, we first present a simplified version of the restart protocol that guarantees resetting the state of every node to a uniform initial state $q_0$.
So, for the time being we may assume that the protocol to be restarted through composition is any Population Protocol $\Pi$ that always starts from the uniform $q_0$ initial configuration (all $u \in V$ in $q_0$ initially).
Later on we shall extend this to handle with protocols that are Network Constructors instead. \\

\noindent \textbf{\textit{Description of the PP Restarting Protocol}}.
The state of every node consists of two components $C_1$ and $C_2$. $C_1$ runs the restart protocol $A$ while $C_2$ runs the given PP $\Pi$.
In general, they run in parallel with the only exception when $A$ restarts $\Pi$.
The $C_1$ component of every node stores a \textit{leader} variable, taking values from $\{l,f\}$, and is initially $l$, a \textit{phase} variable, taking values from $\mathbb{N}_{\geq 0}$, initially $0$, and a \textit{fault} binary flag, initially $0$.

The transition function is as follows. We denote by $x(u)$ the value of variable $x$ of node $u$ and $x'(u)$ the value of it after the transition under consideration.

If a leaders' flag becomes $1$ or $2$, it sets it to $0$, increases its phase by one, and restarts $\Pi$.
If a followers' flag becomes $1$ or $2$, it sets it to $0$, increases its phase by one, becomes a leader, and restarts $\Pi$.
We now distinguish three types of interactions.

When a leader $u$ interacts with a leader $v$, one of them remains leader (state $l$) and the other becomes a follower (state $f$), both set their phase variable to $max\{\text{phase}(u), \text{phase}(v)\}+1$ and both reset their $C_2$ component (protocol $\Pi$) to $q_0$ (i.e., restart $\Pi$).

When a leader $u$ interacts with a follower $v$, if $\text{phase}(u)=\text{phase}(v)$, do nothing in $C_1$ but execute a transition of $\Pi$ (both $u$ and $v$ involved). If $\text{phase}(u) < \text{phase}(v)$, then both set their phase variable to $max\{\text{phase}(u), \text{phase}(v)\}+1$ and both restart $\Pi$, and finally, if $\text{phase}(u) > \text{phase}(v)$, then $\text{phase}'(v) = \text{phase}(u)$ and $v$ restarts $\Pi$.

When a follower $u$ interacts with a follower $v$, if $\text{phase}(u)=\text{phase}(v)$ do nothing in $C_1$ but execute transition of $\Pi$. If $\text{phase}(u)>\text{phase}(v)$, then $v$ sets $\text{phase}'(v)=\text{phase}(u)$ and $v$ restarts $\Pi$, and finally, if $\text{phase}(u)<\text{phase}(v)$, then $u$ sets $\text{phase}'(u)=\text{phase}(v)$ and $u$ restarts $\Pi$. \\

We now show that given any such PP $\Pi$, the above restart protocol $A$ when composed as described with $\Pi$, gives a fault-tolerant version of $\Pi$ (tolerating any number of crash faults).

\begin{lemma}[Leader Election]\label{lemma:leader_election}
	In every execution of $A$, a configuration $C$ with a unique leader is reached, such that no subsequent configuration violates this property.
\end{lemma}
\begin{proof}
	If after the last fault there is still at least one leader, then from that point on at least one more leader appears (due to the fault flags) and only pairwise eliminations can decrease the number of leaders. But pairwise elimination guarantees eventual stabilization to a unique leader.
	It remains to show that there must be at least one leader after the last fault. The leader state becomes absent from the population only when a unique leader crashes. This generates a notification, raising at least one follower's fault flag, thus introducing at least one leader.
\end{proof}

Call a \textit{leader-event} any interaction that changes the number of leaders. Observe that after the last leader-event in an execution there is a stable unique leader $u_l$.

\begin{lemma}[Final Restart]\label{lemma:final_restart}
	On or after the last leader-event, $u_l$ will go to a phase such that $\text{phase}(u_l)>\text{phase}(u),\; \forall u \in V' \setminus \{u_l\}$, where $V'$ denotes the remaining nodes after the crash faults.
	As soon as this happens for the first time, let $S$ denote the set of nodes that have restarted $\Pi$ exactly once on or after that event.
	Then $\forall u \in V' \setminus S, \; u \in S$, an interaction between $u$ and $v$ results in $S \leftarrow S \cup \{u\}$. Thus, $S$ will eventually be $S=V'$.
\end{lemma}
\begin{proof}
	We first show that on or after the last leader-event there will be a configuration in which $\text{phase}(u_l)>\text{phase}(u),\; \forall u \in V' \setminus \{u_l\}$ and it is stable. As there is a unique leader $u_l$ and follower-to-follower interactions do not increase the maximum phase within the followers population, $u_l$ will eventually interact with a node that is in the maximum phase. At that point it will set its phase to that maximum plus one and we can agree that before that follower also sets its own phase during that interaction to the new max, it has been satisfied that $\text{phase}(u_l)>\text{phase}(u),\; \forall u \in V' \setminus \{u_l\}$.
	
	When the above is first satisfied, $S=\{u_l,u\}$ and $\text{phase}(u_l)=\text{phase}(u)>\text{phase}(v),\; \forall v \in V' \setminus S$. Any interaction within $S$, only executes a normal transition of $\Pi$, as in $S$ they are all in the same phase.
	Any interaction between a $u \in V' \setminus S$ and a $v \in S$, results in $S \leftarrow S \cup \{u\}$, because interactions between followers in $V' \setminus S$ cannot increase the maximum phase within $V' \setminus S$, thus $\text{phase}(v)>\text{phase}(u)$ holds and the transition is: $\text{phase}'(u)=\text{phase}(v)$ and $u$ restarts $\Pi$, thus enters $S$.
	It follows that $S$ cannot decrease and any interaction between the two sets increases $S$, thus $S$ eventually becomes equal to $V'$.
\end{proof}

\noindent Putting Lemma \ref{lemma:leader_election} and Lemma \ref{lemma:final_restart} together gives the aforementioned result.

\begin{theorem}
	For any such PP $\Pi$, it holds that $(A,\Pi)$ is a fault-tolerant version of $\Pi$.
\end{theorem}

\begin{lemma}\label{lemma:restart_memory}
	The required memory in each agent for executing protocol $A$ is $O(\log{n})$ bits.
\end{lemma}
\begin{proof}
	Initially all nodes are potential leaders, and they eliminate each other, moving to next phases at the same time. In the worst case, a single leader $u$ will eliminate every other leader, turning them into followers, thus in a failure-free setting the phase of $u$ becomes at most $n-1$.
	If we consider the case where crash faults may occur, each fault can result in notifying the whole population. This will happen if $u$ was adjacent to every other node by the time it crashed. Thus, all nodes increase their phase by one and become leaders again. In the worst case, a single leader eliminates all the other leaders, thus, after the first fault, the maximum phase will be increased by $n-2$.
	The maximum phase than can be reached is $\sum_{i=0}^{k}(n-i) = O(kn)$, where $k$ is the maximum number of faults that may occur ($k<n$).
	 Thus, each node is required to have $O(\log{n})$ bits of memory.
\end{proof}

\noindent \textbf{\textit{NET Restarting Protocol (with Notifications)}}.
We are now extending the \textit{PP Restarting Protocol} in order to handle any NET protocol $\Pi$ (with notifications). Call this new protocol $B$.
We store in the $C_1$ component of each node $u \in V$ a \textit{degree} variable, that is, whenever a connection is formed or deleted, $u$ increases or decreases the value of \textit{degree} by one respectively.
In addition, whenever the \textit{fault flag} of a node $u$ becomes one, it means that an adjacent node of it has crashed, thus it decreases \textit{degree} by one.
In the case of Network Constructors, the nodes cannot instantly restart the protocol $\Pi$ by setting their state to the initial one $q_0$. By Theorem \ref{theorem:global_restart2}, it is evident that we first need to remove all the edges in order to have a successful restart and eventually stabilize to a correct network.

We now define an intermediate phase, called \textit{Restarting Phase} $R$, where the nodes that need to be restarted enter by setting the value of a variable \textit{restart} to $1$ (stored in the $C_1$ component).
As long as their degree is more that zero, they do not apply the rules of the protocol $\Pi$ in their second component $C_2$, but instead they deactivate their edges one by one.
Eventually their degree reaches zero, and then they set \textit{restart} to $0$ and continue executing protocol $\Pi$. We can say that a node $u$, which is in phase $i$ ($\text{phase}(u)=i$), becomes available for interactions of $\Pi$ (in $C_2$) only after a successful restart. This guarantees that a node $u$ will not start executing the protocol $\Pi$ again, unless its degree firstly reaches zero.

The additional Restarting Phase does not interfere with the execution of the \textit{PP Restarting Protocol}, but it only adds a delay on the stabilization time.

\begin{lemma}\label{lemma:degree}
	The variable \textit{degree} of a node $u$ always stores its correct degree.
\end{lemma}
\begin{proof}
	In a failure-free setting, whenever a node $u$ forms a new connection, it increases its \textit{degree} variable by one, and whenever it deactivates a connection, it decreases it by one.
	In case of a fault, all the adjacent nodes are notified, as their \textit{fault flag} becomes one. Thus, they decrease their \textit{degree} by one. In case of a fault with no adjacent nodes, a random node is notified, and its \textit{fault flag} becomes two. In that case, it leaves the value of \textit{degree} the same.
\end{proof}

\begin{theorem}
	For any NET protocol $\Pi$ (with notifications), it holds that $(B,\Pi)$ is a fault-tolerant version of $\Pi$.
\end{theorem}
\begin{proof}
	Consider the case where a node $u$ (either leader or follower) needs to be restarted. It enters to the restarting phase in order to deactivate all of its enabled connections, and it will start executing $\Pi$ only after its degree becomes zero (by Lemma \ref{lemma:degree} this will happen correctly), thus, $\Pi$ always run in nodes with no spurious edges (edges that are the result of previous executions).
	Whenever two connected nodes $u \in R$ and $v \notin R$ interact with each other, they both decrease their \textit{degree} variable by one, and they delete the edge joining them. Obviously, this fact interferes with the execution of $\Pi$ in node $v$ (which is not in the restarting phase), but $v$ is surely in a previous phase than $u$ and will eventually also enter in $R$. This follows from the fact that a node in some phase $i$ can never start forming new edges before it has successfully deleted all of its edges before. New edges are only formed with nodes in the same phase $i$.

	The new \textit{Restarting Phase} does not interfere with the states of the \textit{PP Restarting Protocol}, thus the correctness of $B$ follows by Lemma \ref{lemma:leader_election} and Lemma \ref{lemma:final_restart}.
\end{proof}

\begin{lemma}\label{lemma:restart_memory2}
	The required memory in each agent for executing protocol $B$ is $O(\log{n})$ bits.
\end{lemma}
\begin{proof}
	The maximum value that the variable \textit{degree} can reach is the \textit{maximum reachable degree} (\textit{d}) of protocol $\Pi$. Thus, by Lemma \ref{lemma:restart_memory}, the states that each node is required to have is $O(dkn)$. Both $d$ and $k$ are less that $n-1$, thus, $O(n^3) \; \text{states} = O(\log{n})$ bits.
\end{proof}

\section{Conclusions and Open Problems}\label{conclusions}

A number of interesting problems are left open for future work. Our only exact characterization was achieved in the case of unbounded faults and no notifications. If faults are bounded, non-hereditary languages were proved impossible to construct without notifications but we do not know whether hereditary languages are constructible. Relaxations, such as permitting waste or partial constructibility were shown to enable otherwise impossible transformations, but there is still work to be done to completely characterize these cases. In case of notifications, we managed to obtain fault-tolerant universal constructors, but it is not yet clear whether the assumptions of waste and local coin tossing that we employed are necessary and how they could be dropped. Apart from these immediate technical open problems, some more general related directions are the examination of different types of faults such as random, Byzantine, and communication/edge faults. Finally, a major open front is the examination of fault-tolerant protocols for stable dynamic networks in models stronger than NETs.

\newpage

\bibliographystyle{alpha-abr}
{\normalsize \bibliography{bibliography}}

\end{document}